\documentclass[11pt,a4paper]{article}
\usepackage[T1]{fontenc}
\usepackage[utf8]{inputenc}
\usepackage{authblk}
\usepackage{fullpage}
\usepackage{picture}
\usepackage{todonotes}
\usepackage{hyperref}

\usepackage{tikz}
\usetikzlibrary{arrows, calc, intersections, shapes,positioning, decorations.markings, fadings}
\tikzfading[name=fade out, inner color=transparent!10, outer color=transparent!85]
\tikzset{>=stealth'} 

\usepackage{amsthm,amssymb,amsmath,amscd,amsfonts,latexsym}
\newtheorem{theorem}{{\bf Theorem}}

\newtheorem{definition}[theorem]{{\bf Definition}}

\newtheorem{lemma}[theorem]{{\bf Lemma}}
\newtheorem{remark}[theorem]{{\bf Remark}}

\usepackage{centernot}
\usepackage{tikz}
\usepackage{pgf}
\usepackage{ifthen}

\usetikzlibrary{decorations.pathreplacing,automata,calc,positioning}

\newcommand{\N}{\mathbb{N}}

\newcommand{\Z}{\mathbb{Z}}

\newcommand{\x}{\times}

\newcommand{\R}{Spoiler}
\newcommand{\V}{Duplicator}

\newcommand{\lts}{LTS}

\newcommand{\Act}{A}
\newcommand{\step}[1]{\Step{#1}{}{}}

\newcommand{\wstep}[1]{\Wstep{#1}{}{}}

\newcommand{\Wstep}[3]{\ensuremath{\,{\stackrel{#1}{\Longrightarrow}}\!{}^{\scriptstyle{#2}}_{\scriptstyle{#3}}}\,}

\newcommand{\Step}[3]{\ensuremath{\,{\stackrel{#1}{\longrightarrow}}\!{}^{\scriptstyle{#2}}_{\scriptstyle{#3}}}\,}

\newcommand{\SIM}[2]{\ensuremath{\,\simul^{#1}_{#2}}\,}
\newcommand{\notSIM}[2]{\ensuremath{\,\not\simul^{#1}_{#2}}\,}

\newcommand{\WSIM}{\curlyeqprec}

\newcommand{\PSPACE}{PSPACE}
\newcommand{\pspace}{\PSPACE}
\newcommand{\EXPSPACE}{EXPSPACE}

\newcommand{\OCN}{OCN}
\newcommand{\OCA}{OCA}

\newcommand{\net}{{\cal N}}
\newcommand{\snet}{{\cal S}}
\newcommand{\simul}{\preccurlyeq}
\newcommand{\card}[1]{|#1|}
\newtheorem{claim}{Claim}
\newcommand{\qq}{\text {\sc K}}
\newcommand{\dmax}{d_\text{max}}

\newcommand{\lesssteep}{\prec}

\newcommand{\ignore}[1]{}

\newcommand{\scc}{\it scc}
\newcommand{\acyc}{\it acyc}
\newcommand{\poly}{\it poly}

\newcommand{\suff}[1]{{\it suf}({#1})}

\newcommand{\nic}[1]{}

\newcommand{\cycl}{\text{\sc cycle}}
\newcommand{\pref}{\text{\sc prefix}}

\title{\bf Simulation Over One-counter Nets is PSPACE-Complete
    \thanks{
Technical Report EDI-INF-RR-1418 of the School of Informatics at the University
of Edinburgh, UK. (http://www.inf.ed.ac.uk/publications/report/).
Extended version of material presented at FST\&TCS 2013.
Made available at arXiv.org - Creative Commons License CC-BY.
This work was partially supported by Polish NCN grant 2012/05/NST6/03226 and Polish MNiSW grant N N206 567840.
}}

\author[1]{Piotr Hofman}
\author[1]{S{\l}awomir Lasota}
\author[2]{Richard Mayr}
\author[2]{Patrick Totzke}
\affil[1]{University of Warsaw, Poland}
\affil[2]{University of Edinburgh, UK}

\begin{document}

\maketitle

\begin{abstract}
    One-counter nets (\OCN) are Petri nets with exactly one unbounded place.
They are equivalent to a subclass of one-counter automata with just a weak test for zero.
Unlike many other semantic equivalences, strong and weak simulation preorder
are decidable for \OCN, but the computational complexity was an open problem.
We show that both strong and weak simulation preorder on \OCN\ are \pspace-complete.

\end{abstract}

\section{Introduction}\label{sec:introduction}
{\bf\noindent The model.}
One-counter automata (\OCA) are Minsky counter automata with only one counter,
and they can also be seen as a subclass of pushdown automata with just one
stack symbol (plus a bottom symbol).
One-counter nets (\OCN) are Petri nets with exactly one unbounded place,
and they correspond to a subclass of \OCA\ where the counter
cannot be fully tested for zero, because transitions enabled at counter value
zero are also enabled at nonzero values.
\OCN\ are arguably the simplest model of discrete infinite-state systems,
except for those that do not have 
a global finite control.

\subparagraph*{\bf\noindent Previous results on semantic equivalence checking.}
Notions of behavioral semantic equivalences have been classified in 
Van Glabbeek's linear time - branching time spectrum \cite{Gla2001}.
The most common ones are, in order from finer to coarser, 
bisimulation, simulation and trace equivalence.
Each of these have their standard (called strong) variant, and a weak variant
that abstracts from arbitrarily long sequences of internal actions.

For \OCA/\OCN, strong bisimulation 
is \pspace-complete \cite{BGJ2010}, 
while weak bisimulation is undecidable \cite{May2003}.
Strong trace inclusion is undecidable for \OCA\ \cite{Valiant1973},
and even for \OCN\ \cite{HMT:LICS2013}, and this trivially carries over to weak trace
inclusion.

The picture is more complicated for simulation preorders.
While strong and weak simulation are undecidable for \OCA\ \cite{JMS1999},
they are decidable for \OCN.
Decidability of strong simulation on \OCN\ was first proven in \cite{AC1998},
by establishing that the simulation relation follows a certain regular pattern.
This idea was made more graphically explicit in later proofs \cite{JM1999,JKM2000},
which established the so-called {\em Belt Theorem}, that states that the simulation
preorder relation on \OCN\ can be described by finitely many partitionings of the
grid $\N\x\N$, each induced by two parallel lines.
In particular, this implies that the simulation relation is semilinear.
However, the proofs in \cite{AC1998,JM1999,JKM2000} did not yield any upper
complexity bounds, since the first was based on two semi-decision procedures
and the later proof of the Belt Theorem was non-constructive. 
A \pspace\ lower bound for strong simulation on \OCN\ follows from \cite{Srb2009}.

Decidability of weak simulation on \OCN\ was shown in \cite{HMT:LICS2013}, using 
a converging series of semilinear approximants. This proof used the 
decidability of strong simulation on \OCN\ as an oracle, and thus did not immediately yield any 
upper complexity bound.

\subparagraph*{Our contribution.}
We provide a new constructive proof of the Belt Theorem and derive a \pspace\ algorithm
for checking strong simulation preorder on \OCN.
Together with the lower bound from \cite{Srb2009}, this shows \pspace-completeness of the problem.

Via a technical adaption of the algorithm for weak simulation in \cite{HMT:LICS2013},
and the new \pspace\ algorithm for strong simulation, we also obtain a \pspace\ algorithm
for weak simulation preorder on \OCN. Thus even weak simulation preorder on \OCN\ is 
\pspace-complete.

\newcommand{\undec}{undecidable}
\begin{center}
  \begin{tabular}{ | l | c | c| c | c  | c | }
    \hline
    	 	& simulation 		& bisimulation 			& weak sim.	& weak bis.		 & trace inclusion \\ \hline
    OCN 	& \textbf{\pspace}	& \pspace\ \cite{BGJ2010} 	&\textbf{ \pspace}	& \undec\ \cite{May2003}	 &\undec\ \cite{Valiant1973}\\ \hline
    OCA 	& \undec\ \cite{JMS1999}	& \pspace\ \cite{BGJ2010}	 & \undec\ \cite{JMS1999}	 & \undec\ \cite{May2003} 	&\undec\ \cite{HMT:LICS2013}\\
    \hline
  \end{tabular}
\end{center}

\section{Problem Statement}\label{sec:problem}

A labelled transition system (\lts) over a finite alphabet $\Act$ of actions
consists of a set of configurations and, for every action $a \in \Act$, a binary relation
$\step{a}$ between configurations.

Given two \lts\ $S$ and $S'$, a relation $R$ between the configurations
of $S$ and $S'$ is a \emph{simulation} if for every 
pair of configurations $(c, c') \in R$ and every step $c \step{a} d$
there exists a step $c' \step{a} d'$ such that $(d,d') \in R$.
Simulations are closed under union, so there exists a unique maximal simulation.
If $S=S'$ then this maximal simulation is a preorder, called {\em simulation preorder},
and denoted by $\simul$. If $c \simul c'$ then one says that $c'$ \emph{simulates} $c$.

Simulation preorder can also be characterized by a \emph{Simulation Game} as follows.
The \emph{positions} are
all pairs $(c, c')$ of configurations of $S$ and $S'$ respectively.
The game is played by two players called \emph{\R} and \emph{\V} and proceeds in rounds.
In every round, starting in a position $(c,c')$,
\R\ chooses some $a \in \Act$ and some configuration $d$ with $c \step{a} d$.
Then \V\ responds by choosing a configuration $d'$ with $c' \step{a} d'$,
and the next round continues from position $(d, d')$.
If one of the players cannot move then the other player wins, and \V\ wins
every infinite play.
It is well known that the Simulation Game is determined: 
for every initial position $(c, c')$, exactly one of players has a winning strategy.
Configuration $c'$ simulates $c$ iff \V\ has a strategy to win the Simulation Game from position $(c, c')$.

\begin{definition}[One-Counter Nets]
    A \emph{one-counter net} (\OCN) is a triple $\net=(Q,\Act,\delta)$
    given by finite sets of control-states $Q$, action labels $\Act$
    and transitions $\delta\subseteq Q\x \Act\x\{-1,0,1\}\x Q$.
    It induces an infinite-state labelled transition system over the
    state set $Q\x\N$, whose elements will be written as $pm$, where
    $pm\step{a}qn$ iff $(p,a,d,q)\in\delta \text{ and }n=m+d\ge0$.
\end{definition}
We study the computational complexity of the following decision problem.

\vspace{0.3cm}
\begin{tabular}{ll}
  \multicolumn{2}{l}{\bf Simulation Checking for \OCN}\\
  \hline
  \sc Input:  & Two \OCN\ $\net$ and $\net'$ together with configurations $qn$ and $q'n'$\\
              & of $\net$ and $\net'$ respectively, where $n$ and $n'$ are given in binary.\\
  \sc Question: & $qn \simul q'n'$ ?
\end{tabular}
\vspace{0.3cm}

\begin{theorem}\label{thm:strongsim-pspace}
    The Simulation Checking Problem for \OCN\ is in \pspace.
\end{theorem}
Combined with the \pspace-hardness result of~\cite{Srb2009}, this yields \pspace-completeness of the problem.

\begin{remark}\label{rem:semilinearsize}
Our construction can also be used to compute the simulation relation as a semilinear
set, but its description requires exponential space. 
However, checking a point instance $qn \simul q'n'$ of the simulation problem
can be done in 
polynomial space by stepwise guessing and verifying only a polynomialy bounded part of the
relation; cf.~Section~\ref{sec:pspace}.
\end{remark}

Without restriction (see \cite{AC1998} for a justification) we assume that both
\OCN\ are \emph{normalised}:
\begin{enumerate}
    \item In \R's net $\net$, every control-state has some outgoing transition with a
        non-negative change of counter value.
    \item \V's net $\net'$ is \emph{complete}, i.e., every control-state has an outgoing
        transition for every action (though the change in counter value may be negative).
\end{enumerate}
Thus \R\ cannot get stuck and only loses the game if it is infinite.
Moreover, \V\ can only be stuck (and lose the game) when his counter equals zero.

\subparagraph*{Outline of the proof.}
One easily observes that the Simulation Game is monotone for both players.
If \V\ wins the Simulation Game from a position $(qn, q'n')$ then he also wins from
$(qn, q'm)$ for $m > n'$. Similarly, if \R\ wins from $(qn, q'n')$ then she also wins
from $(qm, q'n')$ for $m > n$.
For a fixed pair $(q,q')$ of control-states, both players winning regions
therefore split the grid $\N\x\N$ into two connected subsets.
It is known \cite{JM1999,JKM2000} that the \emph{frontier} between these subsets
is contained in a \emph{belt}, i.e., it lays between two parallel lines with rational
slope.

For the proof of our main result we analyse a symbolic \emph{Slope Game}. This new
game is similar to the Simulation Game but necessarily ends after a small number of
rounds. We show that given sufficiently high excess of counter-values, both
players can re-use winning strategies for the Slope Game also in the Simulation Game.
As a by-product of this characterization, we obtain polynomial bounds on widths and
slopes of the belts. Once the belt-coefficients are known, one can compute the frontiers
exactly because every frontier necessarily adheres to a regular pattern.

\section{Polynomially Bounded Belts}\label{sec:belt}
Let us fix two \OCN\ $\net$ and $\net'$, with sets of control-states $Q$ and
$Q'$, respectively. Following \cite{JKM2000},
we interpret $\SIM{}{}$ as 2-colouring of $K=|Q\x Q'|$ Euclidean planes, one for each
pair of control-states $(q,q')\in Q\x Q'$.

The main combinatorial insight of \cite{JKM2000}
(this was also present in \cite{AC1998}, albeit less explicitly)
is the so-called \emph{Belt Theorem}, that states that
each such plane can be cut into segments by two parallel lines such that
the colouring of $\SIM{}{}$ in the outer two segments is constant; see Figure~\ref{fig:belt}.
We provide a new constructive proof of this theorem, stated as
Theorem~\ref{thm:belt-theorem} below, that allows us to derive polynomial bounds on
the coefficients of all belts.

\begin{definition}[Positive vectors, direction, c-above, c-below]
    A vector $(\rho,\rho')\in \Z\x\Z$ of integers
    is called \emph{positive} if $(\rho,\rho') \in \N\x\N$ and $(\rho,\rho')\neq (0,0)$.
    Its \emph{direction} is the half-line
    $\mathbb{R}^+\cdot(\rho,\rho')$. 
    For a positive vector $(\rho,\rho')$ 
    and a number $c\in\N$ we say that the point $(n,n')\in\Z\x\Z$ is
    \emph{$c$-above} $(\rho,\rho')$ iff
    there exists some point $(r,r') \in \mathbb{R}^+\cdot(\rho,\rho')$ in the
    direction of $(\rho,\rho')$
    such that
    \begin{equation}
      n < r - c \qquad \text{and} \qquad n' > r' + c.
    \end{equation}
    Symmetrically, $(n,n')$ is \emph{$c$-below} $(\rho,\rho')$ if is a
    point $(r,r')\in \mathbb{R}^+\cdot(\rho,\rho')$ with
    \begin{equation}
      n > r + c \qquad \text{and} \qquad n' < r' - c.
    \end{equation}
\end{definition}

\begin{theorem}[Belt Theorem]\label{thm:belt-theorem}
    For every two one-counter nets $\net$ and $\net'$ with sets of control-states $Q$
    and $Q'$ respectively, there is a bound $c \in \N$ such that 
    for every pair $(q,q')\in Q\x Q'$ of control-states
    there is a positive vector $(\rho, \rho')$ such that
    \begin{enumerate}
        \item if $(n,n')$ is $c$-above $(\rho,\rho')$ then $qn\SIM{}{}q'n'$, and
        \item if $(n,n')$ is $c$-below $(\rho,\rho')$ then $qn\notSIM{}{}q'n'$.
    \end{enumerate}
    Moreover, $c$ and all $\rho,\rho'$ are bounded polynomially w.r.t.~the sizes of $\net$ and $\net'$.
\end{theorem}

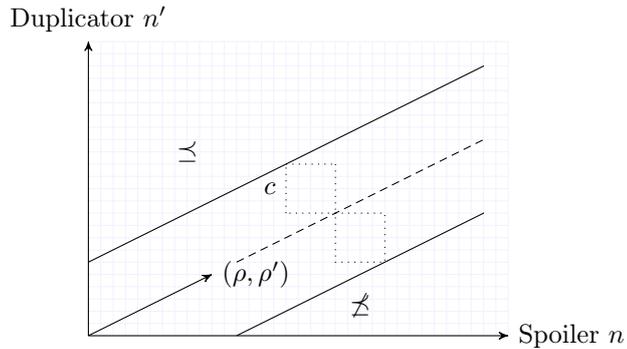
\begin{figure}
  \begin{center}
        \def\c{1}
  \def\RX{5}
  \def\RY{2.5}
  \def\xmax{8}
  \def\ymax{6.5}
  \def\width{8.5}
  \def\height{6}
  \def\xaxispoints{0, 0.5, 1, 1.5, 2, 2.5, 3, 3.5, 4}
  \def\yaxispoints{0, 0.5, 1, 1.5, 2, 2.5, 3, 3.5, 4}
  \def\offset{0}
  \begin{tikzpicture}[scale=0.65]
    \path[use as bounding box](0,0.5) rectangle (8,6);
    \tikzstyle{every node}=[font=\small]
    \coordinate (rhorho') at (\RX,\RY);
    \draw[step=0.25,blue!5!white, very thin] (0,0) grid (\width,\height);
    
    \draw[name path=yaxis, ->] (0,0) -- ($(0,\height)$) node[black,above] {\V\ $n'$};
    \draw[name path=xaxis, ->] (0,0) -- ($(\width,0)$) node[black,right] {\R\ $n$};

    \coordinate (above) at ($(rhorho')+(-\c,\c)$);
    \path[name path=upper] (above) -- +(-\RX, -\RY);
    \path [name intersections={of=upper and yaxis,by=0C}];

    \coordinate (below) at ($(rhorho')+(\c,-\c)$);
    \path[name path=lower] (below) -- +(-\RX, -\RY);
    \path[name intersections={of=lower and xaxis,by=C0}];

    \begin{scope}
      \path[clip, name path=border] (0,0) -- (0,\ymax) -- (\xmax,\ymax) -- (\xmax,0);
      \path[name path=outerborder] (0,\ymax) -- (\xmax,\ymax) -- (\xmax,0);  
      \draw[->] (0,0) -- ($.5*(rhorho')$) node[right]{$(\rho,\rho')$};
      \draw[densely dashed, segment length=20pt] ($.6*(rhorho')$) -- ($10*(rhorho')$);

      \draw[dotted] (rhorho') edge node[below] {} ($(rhorho')+(-\c,0)$);
      \draw[dotted] ($(rhorho')+(-\c,0)$) edge node[left] {$c$} ($(rhorho')+(-\c,\c)$);
      \draw[dotted] ($(rhorho')+(-\c,\c)$) -- ($(rhorho')+(0,\c)$);
      \draw[dotted] ($(rhorho')+(0,\c)$) -- (rhorho');

      \draw[dotted] (rhorho')
                    -- ($(rhorho')+(\c,0)$)
                    -- (below)
                    -- ($(rhorho')+(0,-\c)$)
                    -- cycle;

      \draw[name path=lower] (C0) -- ($(C0)+20*(rhorho')$);  
      \draw[name path=upper] (0C) -- ($(0C)+20*(rhorho')$);  
      
      \path[name intersections={of=upper and outerborder,by=upperEnd}];
      \coordinate (midabove) at ($(upperEnd)!.5!(0,\ymax)$);
      \node at ($(0C)!.5!(midabove)$) {$\preceq$};
      
      \path[name intersections={of=lower and outerborder,by=lowerEnd}];
      \coordinate (midbelow) at ($(lowerEnd)!.5!(\xmax,0)$);
      \node at ($(midbelow)!.5!(C0)$) {$\not\preceq$};
    \end{scope}
  \end{tikzpicture}
  \end{center}
  \caption{A belt with slope $\frac{\rho}{\rho'}$. The dashed half-line is the direction of
      $(\rho,\rho')$.}
  \label{fig:belt}
\end{figure}

%
%

%

\section{Proof of the Belt Theorem}\label{sec:beltproof}
We consider \OCN\ $\net$ and $\net'$ with sets of control-states $Q$ and $Q'$, resp.,
and define the constant $K=|Q\x Q'|$.
Abdulla and Cerans~\cite{AC1998} showed that, above a certain level, the
simulation relation has a regular structure. An important parameter for this 
structure is the {\em ratio} $n/n'$ of the respective counter values $n$ 
in \R's configuration $qn$ of $\net$ and
$n'$ in \V's configuration $q'n'$ of $\net'$.

We further develop this intuition by defining a new finitary game
(called the Slope Game; cf. Section~\ref{subsec:slopegame}) 
that is played directly on the control graphs of the nets,
and in which the objective of the players is to minimize (resp.~maximize)
the ratio of the effects of recently observed minimal cycles.
Then we show how to transform winning strategies in the Slope Game into winning
strategies in the original simulation game.
First we need to define some properties of vectors.

\begin{definition}[Behind, Steeper]
    Let $(\rho,\rho')$ be a positive and $(\alpha,\alpha')\in\Z^2$ an arbitrary vector.
    We place the two on the plane with a common starting point and consider the clockwise oriented angle
    from $(\rho,\rho')$ to $(\alpha,\alpha')$. We say that $(\alpha,\alpha')$ is \emph{behind}
    $(\rho,\rho')$ if the oriented angle is strictly between $0^\circ$ and $180^\circ$. 
    See Figure~\ref{fig:behind} for an illustration.

    Positive vectors may be naturally ordered: We will call $(\rho,\rho')$
    \emph{steeper} than $(\alpha,\alpha')$,
    written $(\alpha,\alpha') \lesssteep (\rho,\rho')$, if $(\alpha,\alpha')$ is behind
    $(\rho,\rho')$.
\end{definition}
Note that the property of one vector being behind another only depends on their directions.
The following simple lemma will be useful in the sequel.
\newpage
\begin{lemma}\label{lem:preserve-above}
Let $(\rho,\rho')$ be a positive vector and $c,n,n'\in \N$.
\begin{enumerate}
    \item If $(n,n')$ is $c$-below $(\rho,\rho')$
        then $(n,n') + (\alpha,\alpha')$ is $c$-below $(\rho,\rho')$
        for any vector $(\alpha,\alpha')$ which is behind $(\rho,\rho')$.
    \item If $(n,n')$ is $c$-above $(\rho,\rho')$
        then $(n,n') + (\alpha,\alpha')$ is $c$-above $(\rho,\rho')$
        for any vector $(\alpha,\alpha')$ which is not behind $(\rho,\rho')$.
\end{enumerate}
\end{lemma}

\begin{figure}[h]
  \begin{minipage}[c]{.475\textwidth}
      \centering
      \includegraphics{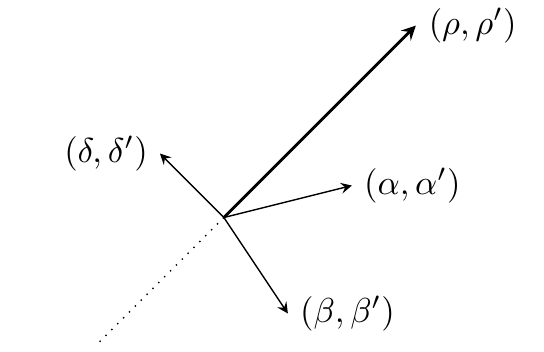}
      \caption{Vectors $(\alpha,\alpha')$ and $(\beta,\beta')$ are behind $(\rho,\rho')$, but
          $(\delta,\delta')$ is not. Also, $(\alpha,\alpha')\lesssteep(\rho,\rho')$.}
  \label{fig:behind}
    \end{minipage}
    \qquad
  \begin{minipage}[c]{.45\textwidth}
  \centering
    \includegraphics{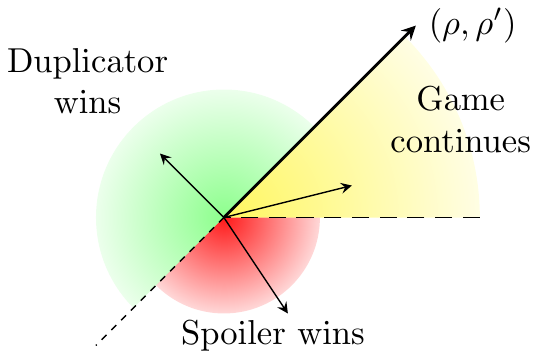}
  \caption{Evaluating the winning condition in position $(\pi,(\rho,\rho'))$
      after a phase of the Slope Game.}
  \label{fig:slope-game}
  \end{minipage}
\end{figure}

\subsection{Slope Game}\label{subsec:slopegame}

\begin{definition}[Product Control Graph, Lasso, Effect of a path]
    Given two \OCN\ $\net=(Q,\Act,\delta)$ and
    $\net'=(Q',\Act,\delta')$,
    their \emph{product control graph} is the finite,
    edge-labelled graph with nodes 
    $Q\x Q'$ and $(\Act\x\N\x\N)$-labelled edges
    $E$ given by
    \begin{equation}
      (p,p')\step{a,d,d'}(q,q')\in E
      \text{ iff } p\step{a,d}q \in\delta
      \text{ and } p'\step{a,d'}q'\in \delta'.
    \end{equation}
    A \emph{path}
    \begin{equation}
        \pi = (q_0,q'_0)\step{a_0,d_0,d_0'}(q_1,q'_1)\step{a_1,d_1,d'_1}\dots\step{a_{k-1},d_{k-1},d'_{k-1}}(q_k,q'_k)
    \end{equation}
    from $(q_0,q'_0)$ to $(q_k,q'_k)$ in this graph
    is called \emph{lasso} if it contains a cycle while none of its strict prefixes does.
    That is, if there exist $i<k$ such that $(q_k,q'_k)=(q_i,q'_i)$ and for all $0\le i<j< k$,
    $(q_i,q'_i)\neq (q_j,q'_j)$. 
    The lasso $\pi$ splits into
    $\pref(\pi)=(q_0,q'_0)\step{a_0,d_0,d_0'}\dots\step{a_{i-1},d_{i-1},d'_{i-1}}(q_i,q'_i)$
    and $\cycl(\pi) = (q_i,q'_i)\step{a_i,d_i,d_i'} \dots\step{a_{k-1},d_{k-1},d'_{k-1}}(q_k,q'_k)$.
    The \emph{effect} of a path is the cumulative sum of the effects of its transitions:
    \begin{equation}
        \Delta(\pi) = \sum_{i=0}^{k-1} (d_i,d_i') \in \Z\x\Z.
    \end{equation}
\end{definition}
The effects of cycles will play a central role in our further construction.
The intuition is that if a play of a Simulation Game describes a lasso then the
players ``agree'' on the chosen cycle. Repeating this cycle will change the ratio of the
counter values towards its effect.

To formalize this intuition, we define a finitary Slope Game which
proceeds in phases.
In each phase, the players alternatingly move on the control graphs of their original
nets, ignoring the counter, and thereby determine the next lasso that occurs.
After such a phase, a winning condition is evaluated that compares the effect of the
chosen lasso's cycle with that of previous phases.
Now either one player immediately wins or the next phase starts, but then the steepness
of the observed effect must have strictly decreased.
The number of different effects of simple cycles thus bounds the maximal length of a game.

\begin{definition}[Slope Game]\label{def:slope_game}
  A \emph{Slope Game} is a strictly alternating two player game played on a pair
  $\net,\net'$ of one-counter nets. The game positions are pairs
  $(\pi,(\rho,\rho'))$, where $\pi$ is an acyclic path in the product control graph
  of $\net$ and $\net'$, and $(\rho,\rho')$ is a positive vector which we call
  \emph{slope}.

The game is divided into \emph{phases}, each starting with a path
$\pi=(q_0,q'_0)$ of length $0$.
Until a phase ends, the game proceeds in rounds like a Simulation Game, but the players pick transition rules instead of transitions:
in a position $(\pi,(\rho,\rho'))$ where $\pi$ ends in states $(q,q')$, \R\ chooses a transition rule
$q \step{a, d} p$,
then \V\ responds with a transition rule $q' \step{a, d} p'$.
If the extended path $\pi'=\pi \step{a,d,d'}(p,p')$ is still not a lasso,
the next round continues from the updated position $(\pi',(\rho,\rho'))$; otherwise the phase ends
with \emph{outcome} $(\pi',(\rho,\rho'))$.
The slope $(\rho,\rho')$ does not restrict the possible moves of either player, nor changes during a phase. 
We thus speak of \emph{the slope of a phase}.

If a round ends in position $(\pi, (\rho, \rho'))$ where $\pi$ is a lasso,
then the winning condition is evaluated.
We distinguish three non-intersecting cases
depending on how the effect $\Delta(\cycl(\pi))=(\alpha,\alpha')$ of the lasso's cycle relates to
$(\rho,\rho')$: 

\begin{enumerate}
  \item If $(\alpha,\alpha')$ is not behind $(\rho, \rho')$, \V\ wins immediately.
  \item If $(\alpha,\alpha')$ is behind $(\rho, \rho')$ but not positive, \R\ wins immediately.
  \item If $(\alpha,\alpha')$ is behind $(\rho, \rho')$ and positive, the game continues with a new
        phase from position $(\pi',(\alpha,\alpha'))$, where $\pi'$ is the path of length $0$
        consisting of the pair of ending states of $\pi$.
\end{enumerate}

Figure~\ref{fig:slope-game} illustrates the winning condition.
Note that if there is no immediate winner it is guaranteed that $(\alpha, \alpha')$ is a positive
vector.
\end{definition}

The fundamental intuition for the connection between the Slope Game and the 
Simulation Game is as follows.
The Slope Game from initial position $((q, q'), (\rho, \rho'))$ determines how
the initial slope $(\rho, \rho')$ relates to the belt in the plane for $(q,q')$
in the simulation relation.
Roughly speaking, if $(\rho,\rho')$ is less steep than the belt then \R\ wins; if $(\rho,\rho')$
is steeper then \V\ wins.
Finally, when the initial slope $(\rho, \rho')$ is exactly as steep as the belt, any player may win the Slope Game.

Consider a Simulation Game in which the ratio $n/n'$ 
of the counter values of \R\ and \V\ is the same as the ratio
$\rho/\rho'$, i.e. suppose $(n,n')$ is contained in the direction of $(\rho,\rho')$.
Suppose also that the values $(n,n')$ are sufficiently large.
By monotonicity, we know that the steeper the slope $(\rho,\rho')$, the better for \V.
Hence if the effect $(\alpha,\alpha')$ of some cycle is behind $(\rho,\rho')$ and positive, 
then it is beneficial for \R\ to repeat this cycle. With more and more repetitions, 
the ratio of the counter values will get arbitrarily close to $(\alpha,\alpha')$.
On the other hand, if $(\alpha,\alpha')$ is behind $(\rho,\rho')$ but not positive then \R\ wins
by repeating the cycle until the \V's counter decreases to $0$.
Finally, if the effect of the cycle is not behind $(\rho,\rho')$ then repeating this cycle leads to \V's win. 

The next lemma follows from the observation that in Slope Games,
the slope of a phase must be strictly less steep than those of all previous
phases.

\begin{lemma}\label{lem:slopegame-bounds}
    For a fixed pair $\net,\net'$ of \OCN,
    \begin{enumerate}
      \item any Slope Game ends after at most $(\qq+1)^2$ phases, and
      \item Slope Games are effectively solvable in \pspace.
    \end{enumerate}
\end{lemma}
\begin{proof}
  After every phase, the slope $(\rho, \rho')$ is equal to the effect of a simple
  cycle, which must be a positive vector.  Thus the absolute values of both numbers
  $\rho$ and $\rho'$ are bounded by $\qq=\card{Q\x Q'}$. It follows that the total
  number of different possible values for $(\rho,\rho')$, and therefore the maximal
  number of phases played, is at most $(\qq + 1)^2$.
  This proves the first part of the claim. Point 2 is a direct consequence
  as one can find and verify winning strategies by an exhaustive search.
\end{proof}

\subparagraph{Strategies in Slope Games.}
Consider one phase of a Slope Game, starting from a position $(\pi, (\rho, \rho'))$.
The phase ends with a lasso whose cycle effect $(\alpha,\alpha')$ satisfies 
exactly one of three conditions, as examined by the evaluating function.
Accordingly, depending on its initial position, 
every phase falls into exactly one of three disjoint cases:

\begin{enumerate}
\item \R\ has a strategy to win the Slope Game immediately,
\item \V\ has a strategy to win the Slope Game immediately or
\item neither \R\ nor \V\ have a strategy to win immediately. 
\end{enumerate}

\noindent
In case 1.~or 2.~we call the phase \emph{final}, and in case 3.~we call it \emph{non-final}.
%
The non-final phases are the most interesting ones because in those,
both players have a strategy that at least prevents an immediate loss.

\subparagraph*{Strategy Trees.}
Both in final and non-final phases, a strategy for \R\ or \V\ is a tree as described below. 
For the definition of strategy trees we need to consider,
not only \R's positions $(\pi, (\rho, \rho'))$ but also \V's positions, the
intermediate positions within a single round. These intermediate positions may be modelled
as triples $(\pi, (\rho,\rho'), t)$ where $t$ is a transition rule in $\net$ from the last state of $\pi$.
Observe that the bipartite directed graph, with positions of a phase as vertices and edges determined by the single-move relation, 
is actually a tree, call it $T$. 
Thus a \R-strategy, i.e.~a subgraph of $T$ containing exactly one successor of every \R's position and all successors
of every \V's position, is a tree as well; and so is any strategy for \V.

Such a strategy (tree) in the Slope Game naturally splits into \emph{segments}, each segment being a strategy (tree) in one phase.
The segments themselves are also arranged into a tree, which we call \emph{segment tree}.
Irrespectively which player wins a Slope Game, according to the above observations,
this player's winning strategy contains segments of two kinds:

\begin{itemize}
  \item non-leaf segments are strategies to either win immediately or continue the Slope
        Game (these are strategies for non-final phases); 
  \item leaf segments are strategies to win the Slope Game immediately (these are strategies in final phases).
\end{itemize}

By the \emph{segment depth} of a strategy we mean the depth of its segment tree.
By Lemma~\ref{lem:slopegame-bounds}, Point 1, we know that a Slope Game ends after at most
$\dmax = (\qq + 1)^2$ phases.
Consequently, the segment depths of strategies are at most $\dmax$ as well.

A value of $c=\qq\cdot\dmax$ is sufficient for the claim of
Theorem~\ref{thm:belt-theorem}.
The intuition behind this value is that for a winning player in the Slope Game,
an excess of $\qq$ per phase is sufficient to be able to safely ``replay'' a
winning strategy in the Simulation Game.
Formally, this is stated by the following two crucial lemmas, proofs of which
can be found in Appendix~\ref{sec:app_proofs}.

\begin{lemma}\label{lem:key-lemma-R}
  Suppose \R\ has a winning strategy of segment depth $d$ in the Slope Game from a
  position $((q,q'), (\rho, \rho'))$. Then \R\ wins the Simulation Game from every
  position $(q n, q' n')$ which is $(\qq \cdot d)$-below $(\rho,\rho')$.
\end{lemma}

\begin{lemma}\label{lem:key-lemma-V}
    Suppose \V\ has a winning strategy of segment depth $d$ in the Slope Game from a
    position $((q,q'), (\rho, \rho'))$. Then \V\ wins the Simulation Game from every
    position $(q n, q' n')$ which is $(\qq \cdot d)$-above $(\rho,\rho')$.
\end{lemma}

\subsection{Proof of Theorem~\ref{thm:belt-theorem}}\label{ssec:belt-theorem-proof}
Let $c=\qq\cdot d_{max}$. For any two states $q\in Q$ and $q'\in Q'$ of the nets
$\net$ and $\net'$ we will determine the ratio $(\rho,\rho')$ that, together with
$c$, characterises the belt of the plane $(q,q')$.
First observe the following monotonicity property of the Slope Game.

\begin{lemma}
  If \R\ wins the Slope Game from a position $((q,q'), (\rho,\rho'))$ and
  $(\sigma,\sigma')$ is less steep than $(\rho,\rho')$ then \R\ also wins the Slope
  Game from $((q,q'), (\sigma,\sigma'))$.
\end{lemma}
\begin{proof}
  Assume that \R\ wins the Slope Game from $((q,q'), (\rho,\rho'))$ while
  \V\ wins from $((q,q'), (\sigma,\sigma'))$, for some $(\sigma,\sigma')\lesssteep (\rho,\rho')$.
  Observe that in both cases, winning strategies of segment depth $\le \dmax$ exist.
  As $(\sigma,\sigma')$ is less steep than $(\rho,\rho')$,
  there is a point $(n,n') \in \N\times\N$ which is both $c$-above $(\sigma,\sigma')$
  and $c$-below $(\rho,\rho')$.
  Applying both Lemma~\ref{lem:key-lemma-R} and~\ref{lem:key-lemma-V}
  immediately yields a contradiction.
\end{proof}
Equivalently, if \V\ wins the Slope Game from $((q,q'), (\rho,\rho'))$ and
$(\sigma,\sigma')$ is steeper than $(\rho,\rho')$ then \V\ also wins the Slope Game
from $((q,q'), (\sigma,\sigma'))$. We conclude that for every pair $(q,q')$ of
states, there is a \emph{boundary slope} $(\beta,\beta')$ such that

\begin{enumerate}
  \item \label{obs:steeper} \R\ wins the Slope Game from $((q,q'), (\sigma,\sigma'))$
      for every $(\sigma,\sigma')$ less steep than $(\beta,\beta')$;
  \item \label{obs:less-steep} \V\ wins the Slope Game from $((q,q'), (\sigma,\sigma'))$
      for every $(\sigma,\sigma')$ steeper than $(\beta,\beta')$.
\end{enumerate}

Note that we claim nothing about the winner from the position $((q,q'), (\beta,\beta'))$ itself.
Applying Lemmas~\ref{lem:key-lemma-R} and~\ref{lem:key-lemma-V} we see that this boundary slope
$(\beta,\beta')$ satisfies the claims 1 and 2 of Theorem~\ref{thm:belt-theorem}.
Indeed, consider a pair $(n,n')\in\N\x\N$ of counter values.
If $(n,n')$ is $c$-below $(\beta,\beta')$, then there is certainly a line $(\bar{\beta},\bar{\beta}')$
less steep than $(\beta,\beta')$ such that $(n,n')$ is $c$-below $(\bar{\beta},\bar{\beta}')$.
By point~\ref{obs:steeper} above, \R\ wins the Slope Game from
$((q,q'), (\bar{\beta},\bar{\beta}'))$. By Lemma~\ref{lem:key-lemma-R}, \R\ wins the
Simulation Game from $(qn,q'n')$.
Analogously, one can use point~\ref{obs:less-steep} above together with Lemma~\ref{lem:key-lemma-V}
to show Point 2 of Theorem~\ref{thm:belt-theorem}.

It remains to show that the boundary slope $(\beta,\beta')$ is polynomial in the sizes of $\net$ and $\net'$.
We show that $(\beta,\beta')$ must in fact be the effect of a simple cycle.
Because such cycles are no longer than $K=|Q\x Q'|$
and because along a path of length $K$ the counter values cannot change
by more than $K$, we conclude that $-K\le \beta,\beta'\le K$.

\begin{definition}[Equivalent vectors]\label{def:vector-equivalence}
Consider all the non-zero effects $(\alpha,\alpha')$ of all cycles together with their
opposite vectors $(-\alpha,-\alpha')$ and denote the set of all these vectors by $V$.
Call two positive vectors $(\rho,\rho')$ and $(\sigma,\sigma')$ \emph{equivalent} if
for all $(\alpha,\alpha') \in V$,
\begin{equation}
    (\alpha,\alpha') \text{ is behind } (\rho,\rho')
    \iff (\alpha,\alpha') \text{ is behind } (\sigma,\sigma').
\end{equation}
\end{definition}
In other words, equivalent vectors lie in the same angle determined by a pair of
vectors from $V$ that are neighbours angle-wise.
We claim that equivalent slopes have the same winner in the Slope Game:

\begin{lemma} \label{lem:constant-winner}
  If $(\rho,\rho')$ and $(\sigma,\sigma')$ are equivalent then the same player wins the Slope Game
  from $((q,q'),(\rho,\rho'))$ and $((q,q'),(\sigma,\sigma'))$.
\end{lemma}
\begin{proof}
  A winning strategy in the Slope Game from $((q,q'),(\rho,\rho'))$ may be literally
  used in the Slope Game from $((q,q'),(\sigma,\sigma'))$.
  This holds because the assumption that $(\rho,\rho')$ and
  $(\sigma,\sigma')$ are equivalent implies that all possible outcomes of the initial phase of the
  Slope Game are evaluated equally.
\end{proof}
Lemma~\ref{lem:constant-winner} implies that the boundary slope is in $V$.
This concludes the proof of Theorem~\ref{thm:belt-theorem}.\qed

\subsection{A Sharper Estimation}\label{sec:sharper:estimation}

Theorem~\ref{thm:belt-theorem} provides a polynomial bound on
the constant $c$ and the slopes of all belts, with respect to the sizes of $\net$ and
$\net'$.
However, the proof of Theorem~\ref{thm:belt-theorem} reveals that a slightly stronger
result actually holds, which will be useful in proving the complexity bound for
weak simulation in Section~\ref{sec:weaksim}.
We can estimate a bound on $c$ in terms of the following two parameters of the
product control graph $\net\x\net'$:
\begin{itemize}
  \item $\text{\sc scc}$, the size of the largest strongly connected component, and
  \item $\text{\sc acyc}$, the length of the longest acyclic path.
\end{itemize}
In particular, we claim that Theorem~\ref{thm:belt-theorem} still holds with
the constant $c$ bounded by
\begin{equation}
    c \leq poly(\text{\sc scc}) + \text{\sc acyc}.
\end{equation}

Intuitively, $c$ is the excess of counter value needed to replay a Slope Game strategy in the Simulation Game.
This directly corresponds to the maximal number of alternations in a play of the Slope Game.
Every phase ends in a cycle, which must be contained in some strongly connected component
and is thus no longer than $\text{\sc scc}$.
So the segment depth of Slope Game strategies is bounded by $(\text{\sc scc}+1)^2$.

We can decompose plays of the Slope Game by separating subpaths that
contain at least one cycle and stay in one strongly connected component, and
the remaining subpaths.
One can now show that in fact, a counter value of $\text{\sc scc}$
suffices to enable subpaths of the first kind. The segment depth bounds
the number of such subpaths in any play.
Secondly, by definition, the subpaths of the second kind cannot share any points.
The sum of their lengths is hence bounded by $\text{\sc acyc}$.
We conclude that a value of $c = (\text{\sc scc}+1)^2 \cdot \text{\sc scc} + \text{\sc acyc}$
is sufficient.

\section{Strong Simulation is PSPACE-complete}\label{sec:pspace}

\newcommand\slope{\text{\sc slope}}
\newcommand\belt{\text{\sc belt}}
\newcommand\abov{\text{\sc above}}
\newcommand\is{L_0}
\newcommand\per{\text{\sc periodic}}
\newcommand \squa[3]{\text{\sc rect}(#1, #2, #3)}
\newcommand\init{\text{\sc init}}
\newcommand\aper{\text{\sc aperiodic}}

Using our stronger version of the Belt Theorem from Section 4,
we derive an algorithm for checking simulation preorder, similarly as in
\cite{AC1998,JM1999,JKM2000}.

As before we fix two \OCN\ $\net$ and $\net'$, with sets of control-states $Q$ and $Q'$, respectively.
By Lemma~\ref{lem:slopegame-bounds}, Point 2, 
we can compute in \pspace, for every pair $(q,q')\in Q \times Q'$, the positive vector $(\rho,\rho')$
satisfying Theorem~\ref{thm:belt-theorem}; we denote this vector by $\slope(q,q')$.
We define $\belt(q,q')$ to be the set of points $(n, n') \in \N^2$ that are neither
$c$-above nor $c$-below $\slope(q,q')$.
As all vectors $\slope(q,q')$ and the widths of all belts are polynomially bounded (by Theorem~\ref{thm:belt-theorem}), 
we observe that every two non-parallel belts are disjoint outside
a polynomially bounded \emph{initial rectangle}, denoted $\is$, between corners
$(0,0)$ and $(l_0, l_0')$ (see Figure~\ref{fig:belts}). 
\begin{figure}[ht]
    \centering
    \def\xmax{8}
  \def\ymax{6.5}
  \def\width{8.5}
  \def\height{6}
  \def\LOx{2cm}
  \def\LOy{1cm}
  \def\pcutn{1.25}
  \def\acutn{2}
  \def\acycColour{red!40}
  \def\periodColour{green!30}
  \def\periodBorderColour{green!20!black}
  \def\initColour{blue!20}

  \begin{tikzpicture}[scale=0.7]
    \path[use as bounding box](-1.5,-0.5) rectangle ($(\width,\height) + (2.5,0.75)$);
    \tikzstyle{every node}=[font=\small]

    \draw[step=0.25,blue!5!white, very thin] (0,0) grid (\width,\height);
    
    \fill[\initColour] (0,0) rectangle (\LOx,\LOy);
    \draw (\LOx,0) -- ($(\LOx,0) -(0,2pt)$) node[below] {\tiny $l_0$};
    \draw (0,\LOy) -- ($(0,\LOy) -(2pt,0)$) node[left] {\tiny $l_0'$};

    \begin{scope}
      \path[clip, name path=border] (0,0) -- (0,\height) -- (\width,\height) -- (\width,0);
      \coordinate (TL) at (-0.75,0);
      \coordinate (BL) at (0.75,0);
      \coordinate (TR) at ($(3.75,0) + (0,\acutn + 2*\pcutn)$);
      \coordinate (BR) at ($(5.25,0) + (0,\acutn + 2*\pcutn)$);
    
      \path[name path=upperbeltline] (TL) -- ($(TL)!5!(TR)$);
      \path[name path=lowerbeltline] (BL) -- ($(BL)!10!(BR)$);

      \path[name path=cutlevel] (0,\acutn) -- (\xmax,\acutn);
      \path[name intersections={of=upperbeltline and cutlevel}];
      \coordinate (ATL) at (intersection-1);
      \path [name intersections={of=lowerbeltline and cutlevel}];
      \coordinate (ATR) at (intersection-1);
      
      \path[name path=cutlevel] (0,\LOy) -- (\xmax,\LOy);
      \path [name intersections={of=upperbeltline and cutlevel}];
      \coordinate (ABL) at (intersection-1);
      \path [name intersections={of=lowerbeltline and cutlevel}];
      \coordinate (ABR) at (intersection-1);
    
      \draw[dotted] (TL) -- (ABL);
      \draw[dotted] (BL) -- (ABR);

      \path[name intersections={of=lowerbeltline and border}];
      \coordinate (MR) at (intersection-1);
      \path[name intersections={of=upperbeltline and border}];
      \coordinate (ML) at (intersection-2);

      \begin{scope}
          \path[clip] (TL) -- (ML) -- (MR) -- (BL);
          \path[fill=\acycColour] (ABL) rectangle (ATR);

          \path[name path=cutlevel] ($(0,\acutn + \pcutn)$) -- ($(\xmax,\acutn + \pcutn)$);
          \path [name intersections={of=lowerbeltline and cutlevel}];
          \coordinate (P1R) at (intersection-1);
          \path [name intersections={of=upperbeltline and cutlevel}];
          \coordinate (P1L) at (intersection-1);
          
          \path[name path=cutlevel] ($(0,\acutn + 2*\pcutn)$) -- ($(\xmax,\acutn + 2*\pcutn)$);
          \path [name intersections={of=lowerbeltline and cutlevel}];
          \coordinate (P2R) at (intersection-1);
          \path [name intersections={of=upperbeltline and cutlevel}];
          \coordinate (P2L) at (intersection-1);
          
          \fill[color=\periodColour,path fading=north] (ATL) rectangle (\xmax,\ymax);

          \path[draw,color=\periodBorderColour] (ATL) -- (ATR);
          \path[draw,color=\periodBorderColour] (P1L) -- (P1R);
          \path[draw,color=\periodBorderColour] (P2L) -- (P2R);
      \end{scope}
      \draw[dotted,path fading=north](BR) -- (MR);
      \draw[dotted,path fading=north](TR) -- (ML);
    \end{scope}

    \draw[->] (ABL) -- ($(ABL)!1.1!(TR)$);
    \draw[->] (ABR) -- ($(ABR)!1.1!(BR)$);

    \draw[decorate,decoration={brace,raise=0.2mm, amplitude=5pt}]
        ($(P1R)+(4mm,0mm)$)
        -- node[auto] {\ \per}
        ($(ATR)+(4mm,0.5mm)$);

    \draw[decorate,decoration={brace,raise=0.2mm, amplitude=5pt}]
        ($(ATR)+(4mm,0mm)$)
        -- node[auto] {\ \aper}
        ($(ABR)+(4mm,0mm)$);

    \node[gray] at ($(ABL)!0.5!(ATR)$) {$A$};
    \node[gray] at ($(ATL)!0.5!(P1R)$) {$P_1$};
    \node[gray] at ($(P1L)!0.5!(P2R)$) {$P_2$};

    \def\alength{1cm}
    \def\plength{1.5cm}
    \def\hbheight{0.75cm}
    \coordinate (ATL) at ($(\LOx,\hbheight)$);
    \coordinate (ABL) at ($(\LOx,0)$);
    \coordinate (ATR) at ($(\LOx+\alength,\hbheight)$);
    \coordinate (ABR) at ($(\LOx+\alength,0)$);
    
    \draw[dotted] (0,\hbheight) -- (ATL);
    \draw[dotted] (0,0) -- (ABL);
    \draw[dotted] (ATR) -- (\width,\hbheight);
    \begin{scope}
        \path[clip] (\LOx,\hbheight) -- (\width,\hbheight) -- (\width,0) -- (\LOx,0);
        \path[fill=\acycColour] (ABL) rectangle (ATR);

        \fill[color=\periodColour,path fading=east] (ABR) rectangle (\xmax,\ymax);
        \path[draw,color=\periodBorderColour] (ATR) -- (ABR);
        \path[draw,color=\periodBorderColour] ($(ATR)+(\plength,0)$) -- ($(ABR)+(\plength,0)$);
        \path[draw,color=\periodBorderColour] ($(ATR)+(2*\plength,0)$) -- ($(ABR)+(2*\plength,0)$);
    \end{scope}
    \draw[->] (\LOx,\hbheight) -- ($(\LOx +5mm + \alength + 2*\plength,\hbheight)$);
  
    \draw (0,0) rectangle (\LOx,\LOy);
    \node[gray] at ($(0,0)!0.5!(\LOx,\LOy)$) {$L_0$};
    
    \draw[name path=yaxis, ->] (0,0) -- ($(0,\height)$) node[black,above] {\V\ $n'$};
    \draw[name path=xaxis, ->] (0,0) -- ($(\width,0)$) node[black,right] {\R\ $n$};

\end{tikzpicture}
\caption{The initial rectangle $\is$ (blue) and two belts.
    Outside $\is$, the colouring of a belt
    consists of some exponentially bounded block (red), and another exponentially bounded non-trivial block (green)
    which repeats ad infinitum along the rest of the belt.
}\label{fig:belts}
\end{figure}
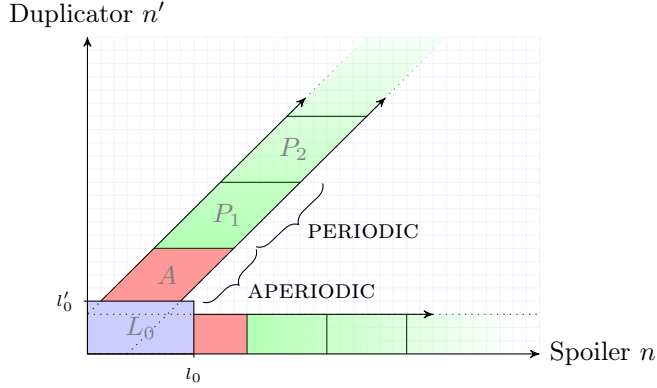

Recall that the simulation preorder on the configurations with the pair of
control-states $(q,q')$ is trivial outside of $\belt(q,q')$:
it contains all pairs $(q n, q' n')$ s.t.\ $(n,n')$ is $c$-above $\slope(q,q')$,
and contains no pairs $(q n, q' n')$ s.t.\ $(n,n')$ is $c$-below $\slope(q,q')$.
We show that inside a belt, the points corresponding to configurations in simulation
are ultimately periodic in the sense defined below.

By the definition of belts, $(n, n') \in \belt(q,q') \iff (n,n') + \slope(q,q') \in \belt(q,q')$, i.e.,
translation via the vector $\slope(q,q')$ preserves membership in $\belt(q,q')$. This
is why we restrict our focus to multiples of vectors $\slope(q,q')$.
We write $\squa{q}{q'}{j}$ for the rectangle between corners $(0,0)$ and $(l_0,l_0') +  j\cdot \slope(q,q')$.

\begin{definition}[ultimately-periodic]
For a fixed pair $(q,q') \in Q\times Q'$ and $j,k \in \N$,
a subset $R \subseteq \belt(q,q')$  is called \emph{$(j,k)$-ultimately-periodic} if
for all $(n,n') \in \N^2 \setminus \squa{q}{q'}{j}$,
\begin{align}
\begin{aligned}
(n,n') \in R \iff (n,n') + k \cdot \slope(q,q') \in R.
\end{aligned}
\end{align}
\end{definition}


\begin{remark}
Observe that for fixed $q$ and $q'$, every $(j,k)$-ultimately-periodic set $R$ can be represented by the numbers
 $j$ and $k$, and two sets
 \[
R  \ \cap \ \squa{q}{q'}{j} \qquad \text{ and } \qquad (R \setminus \squa{q}{q'}{j}) \ \cap \ \squa{q}{q'}{j+k}.
\]
\end{remark}
The following lemma states a property which is crucial for our algorithm.
It is actually a sharpening of the result of~\cite{JKM2000}, with additional effective
bounds on periods inside belts.
\begin{lemma} \label{lem:simul-periodic}
For every pair $(q, q') \in Q\times Q'$, the set
\begin{align*}
\simul_{q,q'} \ \ = \ \ \{ (n,n') \in \belt(q,q') 
: q n \simul q' n' \} 
\end{align*}
is $(j,k)$-ultimately periodic for some $j,k \in \N$ exponentially bounded w.r.t.\ the sizes of $\net$, $\net'$.
\end{lemma}
Thus, when searching for a simulation relation inside belts,
we may safely restrict ourselves to $(j,k)$-ultimately-periodic relations, for exponentially bounded $j$ and $k$.
%
According to the remark above, every such simulation admits the \EXPSPACE\ description that consists, 
for every pair of states $(q,q')$, of:

\begin{itemize}
\item a polynomially bounded vector $(\rho,\rho') = \slope(q,q')$;
\item a polynomially bounded relation $\init(q,q') \subseteq L_0$ inside the initial rectangle $\is$;
\item exponentially bounded natural numbers $j_{q,q'}, k_{q,q'} \in \N$; and
\item two exponentially bounded relations:
\begin{align*}
  \aper(q,q') &\ \subseteq \ \belt(q,q') \ \cap \ \squa{q}{q'}{j_{q,q'}} \\
  \per(q,q') &\ \subseteq \ (\belt(q,q')\setminus \squa{q}{q'}{j_{q,q'}}) \ \cap \  \squa{q}{q'}{j_{q,q'} + k_{q,q'}}.
\end{align*}
\end{itemize}


The above characterization leads to the following naive decision procedure,
which works in \EXPSPACE:
Guess the description of a candidate relation $R$ for the simulation relation, verify
that it is a simulation and check if it contains the input pair of configurations.

Checking whether the input pair is in the (semilinear) relation $R$ is trivial.
To verify that the relation $R$ is a simulation, one needs to check
the \emph{simulation condition} for every pair of configurations $(qn,q'n')$ in $R$,
i.e., \V\ can ensure that after playing one round of the Simulation Game,
the resulting pair of configurations is still in $R$.

The simulation condition is local in the sense that it refers only to positions with
neighbouring counter values (plus/minus $1$).
This, together with the fact that belts are disjoint outside $\is$,
implies that the complete one-neighbourhoods of points in the periodic part repeats along the belt.
It therefore suffices to examine those elements which are in the \EXPSPACE\
description to check if the simulation condition holds.

\subparagraph*{A \pspace\ procedure.}
The naive algorithm outlined above may easily be turned into a \pspace\ algorithm by
a standard shifting window trick.
Instead of guessing the complete exponential-size description upfront, we start by
guessing the polynomially bounded relation inside $\is$ and verifying it locally.
Next, the procedure stepwise guesses parts of the relations $\aper(q,q')$ and later $\per(q,q')$,
inside a polynomially bounded rectangle window through the belt
and shifts this window along the belt, checking the simulation condition for all
contained points on the way.  Since the simulation condition is local, everything
outside this window may be forgotten, save for the
first repetitive window that is used as a certificate for successfully having guessed
a consistent periodic set, once it repeats.
Because this repetition needs to occur after an exponentially bounded number of shifts,
polynomial space is sufficient to store a binary counter
that counts the number of shifts and allows to terminate unsuccessfully once the
limit is reached.
\qed


%

\section{Application to Weak Simulation Checking}\label{sec:weaksim}
A natural extension of simulation is  \emph{weak simulation}, that abstracts from internal steps.
\begin{definition}
For a \lts\ over actions $\Act \cup \{\tau\}$ define
\emph{weak} step relations by
$\wstep{\tau} = \Step{\tau}{*}{}$ and
$\wstep{a} = \Step{\tau}{*}{}\Step{a}{}{}\Step{\tau}{*}{}$ for $a \neq \tau$.
Weak simulation ($\WSIM{}{}$) is now defined just like $\SIM{}{}$, using
Simulation Games, in which \V\ moves along weak steps.
\end{definition}

For systems without $\tau$-labelled transitions, $\step{a} = \wstep{a}$ and therefore strong and weak simulation coincide.
The \pspace\ lower bound from \cite{Srb2009} for checking strong
simulation thus also holds for weak simulation checking over \OCN.

Weak simulation has recently been shown to be decidable for \OCN~\cite{HMT:LICS2013}.
The main obstacle was that \V's system is infinitely branching
w.r.t.\ the weak $\wstep{a}$ steps, 
which implies that non-simulation does not
necessarily manifest itself locally.

In \cite{HMT:LICS2013}, this problem is resolved by
constructing a monotone decreasing sequence of semilinear \emph{approximant relations}
that converges to weak simulation at a finite index.
The approximant relations are derived from a symbolic 
characterization of \V's infinitely-branching system.
They can be computed inductively
by characterizing them in terms of strong simulation over suitably modified \OCN.
The fact that one can effectively compute semilinear descriptions
of $\SIM{}{}$ over \OCN\ \cite{JKM2000} allows to successively compute
the approximant relations and to detect convergence of the sequence.

Here we show that the polynomial bounds from Theorem~\ref{thm:belt-theorem},
together with the technique from \cite{HMT:LICS2013}, imply a \pspace\
upper bound even for checking \emph{weak} simulation on \OCN.
In particular, we claim that the sizes of the ``suitably modified \OCN''
mentioned above, which characterize the approximants, 
are in fact polynomial for every index $i\in\N$ in the sequence.
A more detailed analysis can be found in Appendix~\ref{sec:app_weaksim}.

\begin{theorem}\label{thm:weakpspace}
Checking weak simulation preorder on \OCN\ is \pspace-complete.
\end{theorem}

\section{Conclusion}\label{sec:conclusion}
We have shown that both strong and weak simulation preorder checking between
two given \OCN\ processes is \pspace-complete. Moreover, it is possible to
compute representations of the entire simulation relations as semilinear sets,
but these require exponential space. One cannot expect polynomial-size 
representations of the relations as semilinear sets, because otherwise one
could first guess the representation and then verify in ${\it coNP}^{\it NP}$ 
(for strong simulation) that there are
no counterexamples to the local simulation condition. This would yield an algorithm
in $\Sigma_p^3$ in the polynomial hierarchy, which (under standard assumptions
in complexity theory) contradicts the \pspace-hardness of the problem.




\appendix
\newpage
\section{Missing Proofs from Sections \ref{sec:beltproof} and \ref{sec:pspace}}
\label{sec:app_proofs}

\subsection{Proof of Lemma~\ref{lem:key-lemma-R}}  \label{sec:proof-R}

Suppose \R\ wins the Slope Game from $((q,q'), (\rho, \rho'))$ using a strategy of segment depth $d$.
A position in the Slope Game contains a positive vector $(\rho,\rho')$,
while a position in the Simulation Game contains a pair $(n,n') \in \N\x\N$ 
of counter values, that can also be interpreted as a positive vector.
We will derive a strategy for \R\ in the Simulation Game that is winning from all positions
$(qn,q'n')$ where $(n,n')$ is $(\qq \cdot d)$-below $(\rho,\rho')$.
%
%
%
The crucial idea of the proof is to consider the segments of the supposed winning strategy in the
Slope Game separately.
Each such segment is a strategy for one phase and as such, describes how to move in
the Simulation Game until the next lasso is observed. Afterwards, \R\ can
chose to continue playing according to the next lower segment, or ``roll back'' the
cycle and continue playing according to the current segment.
By the rules of the Slope Game we observe that after sufficiently many such rollbacks
the difference between the ratio $n/n'$ of the actual counters and the slope of the next
lower segment is negligible, i.e., these vectors are equivalent in the sense
of Definition~\ref{def:vector-equivalence} in Section~\ref{ssec:belt-theorem-proof}.
Then, \R\ can safely continue to play according to the next lower segment at level
$d-1$.

To safely play such a strategy in the Simulation Game, \R\ needs to ensure that her own
counter does not decrease too much as that could restrict her ability to move.
We observe however, that any partial play that ``stays in some segment'' at height $d$,
can be decomposed into a single acyclic prefix plus a number of cycles.
Such a play therefore preserves the invariant that all visited points are
$\qq\cdot (d-1)$-below the slope of the phase. In particular, this means that \R's
counter is always $\ge\qq\cdot(d-1)$.

Formally, the proof of Lemma~\ref{lem:key-lemma-R} proceeds by induction on the segment depth $d$.

\subparagraph*{Case $d = 1$.} 
This means that \R\ has a strategy to win the Slope Game in the first phase,
and hence to enforce that the effect of all cycles is behind $(\rho,\rho')$ but not
positive. Denote this strategy by $T$.
In the Simulation Game, \R\ will re-use this strategy as we describe below.
At every position $(qn,q'n')$ in the Simulation Game, 
\R\ keeps a record of the \emph{corresponding position} $(\pi, (\rho,\rho'))$ in the Slope Game, enforcing the
invariant that $(q,q')$ are the ending states of the path $\pi$.

From the initial position  $(qn,q'n')$ with corresponding position $((q,q'), (\rho,\rho'))$,
\R\ starts playing the Simulation Game according to $T$,
until the path in the corresponding position of the Slope Game, say $\pi_1$,
describes a lasso (this must happen after at most $\qq$ rounds).
Thus $\pi_1$ splits into:
\begin{equation}
  \pi_1 = \widetilde\pi_1 \, \bar\pi_1
\end{equation}
where the suffix $\bar\pi_1$ is a  cycle. Denote by $(\widetilde\alpha_1, \widetilde\alpha'_1)$
and $(\bar\alpha_1, \bar\alpha'_1)$ the effects of $\widetilde\pi_1$ and
$\bar\pi_1$, respectively. 
The current values of counters are clearly
\begin{equation}
 n + \widetilde\alpha_1 + \bar\alpha_1 \qquad \text{and }\quad
 n' + \widetilde\alpha'_1 + \bar\alpha'_1
\end{equation}
assuming that the play did not end by now with \R's win.
As the length of path $\pi_1$ is at most $\qq$ and
$(n,n')$ is assumed to be $\qq$-below $(\rho,\rho')$, we know that
all positions visited by now in the Simulation Game were below $(\rho,\rho')$.
In particular, \R's counter value was surely non-negative by now.

Now \R\ \emph{rolls back} the cycle $\bar\pi_1$, namely changes the corresponding
position in the Slope Game from $(\pi_1, (\rho,\rho'))$ to $(\widetilde\pi_1,(\rho,\rho'))$ 
and continues playing according to $T$. 
The play continues until \R\ wins or the path in the corresponding position of the
Slope Game, say $\pi_2$, is a lasso again. Again, we split the path into an acyclic
prefix and a cycle:
\begin{equation}
  \pi_2 = \widetilde\pi_2 \, \bar\pi_2.
\end{equation}
Denote the respective effects by $(\widetilde\alpha_2, \widetilde\alpha'_2)$
and $(\bar\alpha_2, \bar\alpha'_2)$.
A crucial but simple observation is that, assuming that the play did not end by now with \R's win,
the current values of counters are now
\begin{equation}
    n + \widetilde\alpha_2 + \bar\alpha_1 + \bar\alpha_2 \qquad \text{and }\quad
    n' + \widetilde\alpha'_2 + \bar\alpha'_1 + \bar\alpha'_2,
\end{equation}
i.e. the effect $(\widetilde\alpha_1, \widetilde\alpha'_1)$ of $\widetilde\pi_1$ does not contribute any more.
As $(\bar\alpha_1, \bar\alpha'_1)$ is behind $(\rho,\rho')$ we may
apply Lemma~\ref{lem:preserve-above} to $(\bar\alpha_1, \bar\alpha'_1)$ and $c = 0$
in order to deduce, similarly as before, that all positions by now were below $(\rho,\rho')$.
Now \R\ rolls back $\bar\pi_2$ by establishing $(\widetilde\pi_2, (\rho,\rho'))$ as
the new corresponding position in the Slope Game.
Continuing in this way, after $k$ rollbacks the counter values are:
\begin{align}
\begin{aligned}
 &n\ + \widetilde\alpha_k + (\bar\alpha_1 + \bar\alpha_2 + \ldots \ + \bar\alpha_{k-1}) + \bar\alpha_k
 \qquad \text{and}\\
 &n' + \widetilde\alpha'_k + (\bar\alpha'_1 + \bar\alpha'_2 + \ldots + \bar\alpha'_{k-1}) + \bar\alpha'_k,
 \end{aligned}
\end{align}
assuming that \R\ did not win earlier. 
All the vectors $(\bar\alpha_i, \bar\alpha'_i)$, and thus also the sum
\begin{equation}\label{eq:vectorsum}
  (\bar\alpha_1 + \bar\alpha_2 + \ldots \ + \bar\alpha_{k-1}, \bar\alpha'_1 + \bar\alpha'_2 + \ldots + \bar\alpha'_{k-1})
\end{equation}
are behind $(\rho,\rho')$, hence similarly as before all positions by now have been below $(\rho,\rho')$,
by Lemma~\ref{lem:preserve-above} applied to the vector~\eqref{eq:vectorsum} above.

This in particular means that \R's counter remains above value $c$.
However, as by assumption all observed cycles come from a final segment in her Slope Game strategy,
the vector~\eqref{eq:vectorsum} cannot be positive for any $k$. Thus, every rollback strictly
decreases \V's counter value.
We conclude that after sufficiently many rollbacks, \V's counter must eventually drop below $0$ and
hence \R\ eventually wins.
\subparagraph*{Case $d > 1$.} 
By assumption, \R\ has a strategy $T$ for the Slope Game, which has segment depth $d>0$.
As before, \R's strategy in the Simulation Game will re-use the strategy $T$
from the Slope Game, using rollbacks.

\R\ plays according to the initial segment of this strategy, that allows her to win or at least
guarantee that the effect 
of the first observed lasso's circle is less steep than $(\rho,\rho')$.
After $l$ rollbacks, the counter values will be of the form:
\begin{align} \label{eq:countervalues}
\begin{aligned}
  &n + \widetilde\alpha + (\bar\alpha_1 + \ldots \ + \bar\alpha_{m})  + (\bar\gamma_1 +\ldots \ + \bar\gamma_{l})\quad\text{and}\\
  &n' + \widetilde\alpha' + (\bar\alpha'_1 + \ldots + \bar\alpha'_{m}) + (\bar\gamma'_1 + \ldots + \bar\gamma'_{l}),
 \end{aligned}
\end{align}
where the absolute values of $\widetilde\alpha$ and $\widetilde\alpha'$ are at most $\qq$,
the vectors $(\bar\gamma_i,\bar\gamma'_i)$ are behind $(\rho,\rho')$ and positive,
and the vectors $(\bar\alpha_i,\bar\alpha'_i)$ are behind $(\rho,\rho')$ and non-positive.
We apply Lemma~\ref{lem:preserve-above} to $c = \qq\cdot(d-1)$ and learn
that all the positions by now have been $(\qq\cdot(d-1))$-below $(\rho,\rho')$.

In general \R\ has no power to choose whether a effect of a cycle at a next rollback is positive or not.
However, if from some point on all effects are non-positive then \V's counter
eventually drops below $0$ and \R\ wins.
Thus w.l.o.g\., we focus on positions in the Simulation Game immediately after a rollback of a
cycle with positive effect.
Using the notation from~\eqref{eq:countervalues}, suppose $(\gamma_l,\gamma'_l)$ 
is the effect of the last rolled back cycle.
We need the following claim in order to apply the induction assumption:
\begin{claim}
  After sufficiently many rollbacks the vector $(\bar n, \bar n')$ of current counter values
  in the Simulation Game is $(\qq\cdot(d-1))$-below some vector $(\gamma,\gamma')$
  which is equivalent to the positive effect $(\gamma_l,\gamma'_l)$ of the last rolled back cycle.
\end{claim}
\begin{proof}
  By an easy geometric argument. Ignore vectors $(\alpha_i,\alpha'_i)$ as they
  preserve being $(\qq\cdot(d-1))$-below all positive vectors that are less steep
  than $(\rho,\rho')$.  If \V\ wants to falsify the condition, he would need to
  increase the steepness of the rolled back cycle infinitely often, which is
  clearly impossible as there are only finitely many simple cycles.
\end{proof}

Let $(\bar q \bar n, \bar q' \bar n')$ be a position of the Simulation Game satisfying the claim.
We know that \R\ has a winning strategy in the Slope Game from $((\bar q, \bar q'),
(\gamma_l,\gamma'_l))$, of segment depth at most $d-1$.
Because $(\gamma_l,\gamma'_l)$ is equivalent to $(\gamma,\gamma')$, we can apply
Lemma~\ref{lem:constant-winner} and know that the same strategy is winning in the
Slope Game from $((\bar q, \bar q'), (\gamma, \gamma'))$.
By the induction assumption we conclude that \R\ wins the Simulation Game
from $(\bar q \bar n, \bar q' \bar n')$, which completes the proof of
Lemma~\ref{lem:key-lemma-R}.\qed

\subsection{Proof of Lemma~\ref{lem:key-lemma-V}}  \label{sec:proof-V}

Suppose \V\ wins the Slope Game from $((q,q'), (\rho, \rho'))$ using a strategy of segment depth $d$.
We will show that \V\ wins the Simulation Game from every position $(q n, q' n')$ where
$(n,n')$ is $(\qq \cdot d)$-above $(\rho,\rho')$.
We will again build on the concept of rollbacks and proceed by induction on $d$.

\subparagraph*{Case $d = 1$.} 
In this case, \V\ has a strategy to win the Slope Game immediately after the first phase.
This means he can enforce that the effects of the cycles of all observed lassos are not behind
$(\rho,\rho')$.
By a straightforward induction using part 2 of Lemma~\ref{lem:preserve-above}
one can show that \V\ can preserve the invariant that all visited points
are $K$-above $(\rho,\rho')$. This in particular means that
his counter value stays positive and he wins by enforcing an infinite play.

\subparagraph*{Case $d > 1$.} 
Let $T$ denote the initial segment of \V's strategy in the Slope Game.
Every effect of a cycle in $T$ is either not behind $(\rho,\rho')$ or behind
$(\rho,\rho')$, but positive.

In the Simulation Game, \V\ will play according to this initial segment $T$, using
rollbacks, as long as the effect of the rolled back cycle is not behind
$(\rho,\rho')$.
Just as in the previous case, we can apply part 2 of Lemma~\ref{lem:preserve-above}
for $c = \qq\cdot d$ and derive that in this way, \V\ is able to keep the current
counter values $(\qq\cdot d)$-above $(\rho,\rho')$.
 
Suppose that after a few iterations, the effect $(\alpha,\alpha')$ of 
the last cycle \emph{is} behind $(\rho,\rho')$ and let $(\bar q \bar n, \bar q' \bar n')$ be the
position in the Simulation Game directly afterwards.
In this case, $(\alpha, \alpha')$ is clearly positive and less steep than
$(\rho,\rho')$.
Now the point described by the counter values before this last cycle was $(\qq\cdot
d)$-above $(\rho,\rho')$ and because the cycle is no longer than $K$, we know that the
point $(\bar n, \bar n')$ of current counter values (after the cycle) is still
$(\qq\cdot(d-1))$-above $(\rho,\rho')$.
This means, as $(\alpha, \alpha')\lesssteep(\rho,\rho')$,
that $(\bar n, \bar n')$ is also $(\qq\cdot(d-1))$-above $(\alpha,\alpha')$.

Knowing that \V\ has a winning strategy in the Slope Game from $((\bar q, \bar q'),
(\alpha,\alpha'))$ of segment depth at most $d-1$, by induction assumption we obtain 
a winning strategy for \V\ in the Simulation Game from $(\bar q \bar n, \bar q' \bar n')$. 
This completes the description of \V's winning strategy from $(q n, q' n')$ and hence also
the proof of Lemma~\ref{lem:key-lemma-V}.\qed

\subsection{Proof of Lemma~\ref{lem:simul-periodic}}\label{proof:lemma:x}
%
%
For technical convenience we assume w.l.o.g.~that no belt contains the upper right corner of $\is$
(this can always be achieved by minimally extending $\is$, if necessary.)
Thus every belt intersects either the horizontal, or the vertical border of $\is$, but not both.

Recall that the non-parallel belts do not overlap/interfere with each other outside $\is$, hence 
we can consider them separately. 
For the rest of the proof fix states $q, q'$ and let $(\rho, \rho') = \slope(q,q')$.
W.l.o.g.~suppose that $\belt(q,q')$ intersects the horizontal border of $\is$ 
(if it intersects the vertical border of $\is$ the proof is analogous).

For simplicity we assume that no other belt is parallel to $\belt(q,q')$.
The proof below may be easily adapted to the general case by considering a bunch of parallel belts
jointly, instead of just the single one $\belt(q,q')$.

By a \emph{cross-section} at level $n'$ we mean the intersection of $\simul_{q,q'}$
with two consecutive horizontal lines at that level,
i.e.~with $\{ (n,n'), (n,n'+1) : n \in \N \}$.
We may assume that cross-sections are always non-empty (this can always be ensured by
slightly widening $\belt(q,q')$ if necessary).
We say that two cross-sections $s_1$ and $s_2$ are \emph{equal} if
one of them is obtained by a shift of the other by a multiple of $(\rho, \rho')$; formally, we require 
for some $k \in \N$,
\begin{align} \label{eq:equalcs}
s_1 + k \cdot (\rho, \rho') \  \ = \ \ s_2.
\end{align}
%
Choose two cross-sections $s_1, s_2$ at levels $n'_1$ and $n'_2$ respectively,
and $k > 0$ that satisfies~\eqref{eq:equalcs}.
Let $P$ be the restriction of $\simul_{q,q'}$ to the area between $s_1$ and $s_2$, 
and $A$ be the restriction of $\simul_{q,q'}$ to the area below $s_1$:
\begin{align*}
A \ & = \  \{ (n,n') \in \  \simul_{q,q'} \ : \ n' < n'_1 \} &
P \ & = \ \{ (n,n') \in \ \simul_{q,q'} \ : \ n'_1 \leq n' < n'_2 \}.
\end{align*}
Recall that $A$ and $P$, similarly as $\simul_{q,q'}$, are subsets of $\belt(q,q')$. 
We claim: 
\begin{lemma} \label{lem:periodicbelt}
For every $s_1, s_2$ and $k > 0$ satisfying~\eqref{eq:equalcs},
\[
\simul_{q,q'} \ = \ A \ \cup \ P^*, \quad \text{ where } P^* \ = \ \bigcup_{i \in \N} (P + i \cdot k \cdot (\rho,\rho')).
\]
\end{lemma}
Before proving this lemma note that it implies Lemma~\ref{lem:simul-periodic}.
Indeed, by Theorem~\ref{thm:belt-theorem}, a cross-section contains polynomially many points, and therefore there are at most exponentially many non-equal
cross sections.
Thus, by the pigeonhole principle, there are surely two equal cross-sections
at exponentially bounded levels $n'_1$ and $n'_2$.
 
%

Now we prove Lemma~\ref{lem:periodicbelt}. The proof strongly relies on the locality of the simulation condition.
We first claim one inclusion of Lemma~\ref{lem:periodicbelt}, namely:
\begin{claim} \label{cl:incl}
$A \ \cup \ P^* \subseteq \ \simul_{q,q'}$.
\end{claim}
\begin{proof}
We show  that the following relation is a simulation:
\[
R \ \ = \ \ \simul \ \setminus \ \{ (q n, q' n') : (n,n') \in \belt(q,q')\}
\ \cup \ \{ (q n , q' n') : (n,n') \in A \ \cup \ P^* \}.
\]
(Roughly speaking, $R$ is obtained from $\simul$ by replacing $\simul_{q,q'}$ with $A \ \cup \ P^*$.)
We claim that $R$ is a simulation, relying on the locality of the simulation condition.
Formally, we define the \emph{relative $R$-neighborhood} of a point $(n,n')$ as 
\[
\{ (p l, p' l') : (p (n+l), p' (n'+l')) \in R, \ (p, p') \in Q\times Q', \ l, l' \in \{-1,0,1\} \} .
\] 
Note that the simulation condition for a pair of configurations $(q n, q' n')$ with respect to the relation $R$ 
only depends on the relative $R$-neighborhood of $(n, n')$. 
Similarly, one defines the relative $\simul$-neighborhood of a point $(n, n')$.

By the definition of cross-section and of the sets $A$ and $P$, 
the relative $R$-neighborhood of a point $(n,n') \in R$ equals the relative $\simul$-neighborhood of some 
(possibly other) point in $\simul_{q,q'}$.
Thus we deduce that every pair in $R$ satisfies the simulation condition wrt.~$R$, i.e.~$R$ is a simulation.
As $\simul$ is the largest simulation, the claim follows.
\end{proof}

%
%
In order to show the other inclusion of Lemma~\ref{lem:periodicbelt}, extend $n'_1$ and $n'_2$ to an infinite arithmetic progression
\[
n'_1, \ n'_2, \ n'_3, \ \ldots,
\]
i.e.~$n_{i+1} = n'_i + k \cdot \rho'$ for $i \geq 1$,
and consider the ``segments'' $P_i$ of $\simul_{q,q'}$ defined by the corresponding cross-sections:
\[
P_i \ = \  \{ (n,n')  \in \ \simul_{q,q'} \ : \ n'_i \leq n' < n'_{i+1} \} \qquad \text{ for } i \geq 1. 
\]
Clearly, $P = P_1$ and 
$\simul_{q,q'} \ = \ A \ \cup \ \bigcup_{i \geq 1} P_i$.
By Claim~\ref{cl:incl} it follows that $P_1 + k \cdot (\rho,\rho') \subseteq P_2$, or equivalently
$P_1 \subseteq P_2 - k \cdot (\rho,\rho')$. Analogously one shows:
\begin{align}  \label{eq:incl}
P_i  \ \subseteq \ P_{i+1} - k\cdot (\rho,\rho') \qquad \text{ for every } i \geq 1.
\end{align}
We claim that the inclusions are actually equalities:

\begin{claim} \label{cl:eq}
$P_i \ = \ P_{i+1} - k \cdot (\rho, \rho')$, for every $i \geq 1$.
\end{claim}
\begin{proof}
Due to Equation~\eqref{eq:incl}, it suffices to show the inclusions $P_{i+1} - k \cdot (\rho, \rho') \ \subseteq \ P_i$.
The inclusions follow, similarly as in the proof of Claim~\ref{cl:incl}, from the observation that the following relation
is a simulation:
\[
R \ \ =  \ \ \simul \ \setminus\  \{ (q n, q' n') : (n, n') \in \belt(q,q') \} 
\ \ \cup \ \ \{ (q n, q' n') : (n,n') \in A 
\ \cup \ \bigcup_{i \geq 2} P_i - k \cdot (\rho,\rho') \}.
\]
The relation $R$ is obtained from $\simul$, roughly speaking, by removing the first segment $P_1$ and shifting all other segments
$P_i$ by vector $- k \cdot (\rho, \rho')$. To prove that $R$ is a simulation, we exploit locality of the simulation condition 
exactly as before.
Additionally, we use the observation that the simulation condition is monotonic with respect to inclusion of relative neighborhoods,
together with the inclusions~\eqref{eq:incl}.
\end{proof}

\noindent
Claim~\ref{cl:eq} immediately implies Lemma~\ref{lem:periodicbelt}
and thus Lemma~\ref{lem:simul-periodic}.

\newpage
\section{Weak Simulation Checking}
\label{sec:app_weaksim}
We show that the bounds on the coefficients of the Belt Theorem,
as derived in Section~\ref{sec:beltproof}, imply that the construction
from \cite{HMT:LICS2013} for checking \emph{weak} simulation uses only polynomial space.

In order to avoid repeating the involved construction
from \cite{HMT:LICS2013}, we refer the reader to the original paper for technical details
and recover here only those notions and properties which suffice
to provide some intuition and derive the claimed \pspace\ bound.

We aim to compute a description of $\WSIM{}{}$, the largest weak simulation over a given pair
of \OCN. First we reduce this weak simulation game to
a strong simulation game between two modified systems.

\newcommand{\mnet}{{\cal M}}

\begin{definition}[$\omega$-Nets]
    An \emph{$\omega$-net} $\mnet=(Q,\Act,\delta)$ is given by a finite set of
    control-states $Q$, a finite set of actions $\Act$ and transitions
    $\delta\subseteq Q\x \Act\x\{-1,0,1,\omega\}\x Q$. It induces a transition system over the
    stateset $Q\x\N$
    that allows a step $pm\step{a}p'm'$ if either $(p,a,d,p')\in\delta$ and $m'=m+d\in\N$ or if
    $(p,a,\omega,p')\in\delta$ and $m' > m$.
\end{definition}

\begin{lemma}[\cite{HMT:LICS2013}]\label{lem:reduction}
    For two \OCN\ $\net$ and $\net'$ with sets of control-states $Q$ and $Q'$ resp., one
    can construct a \OCN\ $\mnet$ with states $Q_{\mnet}\supseteq Q$
    and an $\omega$-net $\mnet'$ with states $Q_{\mnet'}\supseteq Q'$, such that for
    each pair $(q,q')\in Q\x Q'$ of original control states,
    \begin{equation}
        qn\WSIM{}{} q'n'\text{ w.r.t. }\net,\net'\text{ iff }qn\SIM{}{} q'n'\text{ w.r.t. }\mnet,\mnet'.
    \end{equation}
    Moreover, the sizes of $\mnet$ and $\mnet'$ are polynomial
    in the size of $\net$ and $\net'$.
\end{lemma}

Thus, it suffices to compute a description of the strong simulation relation relative to
a given \OCN\ $\mnet=(Q,\Act,\delta)$ and an $\omega$-net $\mnet'=(Q',\Act,\delta')$.
To do that, we construct a sequence of successively decreasing (w.r.t.~set inclusion)
approximant relations $\SIM{i}{}$ and show that 1) for all $i\in\N$, $\SIM{i}{}$ is effectively semilinear and 2)
there is some $k\in\N$ with $\SIM{k}{}=\SIM{k+1}{}=\SIM{}{}$, i.e., the sequence converges to
simulation preorder at some finite level $k$.

Intuitively, $\SIM{i}{}$ is given by a \emph{parameterized simulation game} that keeps track of
how often \V\ uses $\omega$-labelled transitions and in which
\V\ immediately wins if he plays such a step the $i$th time.
It is easy to see that this game favours \V\ due to the additional winning condition.
With growing index $i$, this advantage becomes less important and the game
increasingly resembles a standard simulation game. Hence,
$\forall i\in\N, \SIM{i}{}\supseteq\SIM{i+1}{}$.

In \cite{HMT:LICS2013}, it is shown that these approximants $\SIM{i}{}$ can in fact be characterized 
by equivalent (in the sense of Lemma~\ref{lem:snets} below) ordinary strong simulation 
relations between suitably extended \OCN.
\begin{lemma}\label{lem:snets}
There is a sequence $(\snet_i, \snet_i')$ of pairs of \OCN\ such that for all indices $i\in\N$:
\begin{enumerate}
    \item $\snet_i$ and $\snet_i'$ contain all states of $\mnet$ and $\mnet'$ respectively.
    \item \label{lem:snets:char} For all configurations $qn\in (Q\x\N)$ and $q'n'\in (Q'\x\N)$ of $\mnet$ and $\mnet'$ it holds
          that $qn\SIM{i}{}q'n'$ w.r.t.~$\mnet,\mnet'$ iff $qn\SIM{}{}q'n'$ w.r.t.~$S_i,S_i'$.
    \item $\snet_{i+1}$ and $\snet_{i+1}'$ can be computed from $\snet_i$ and $\snet_i'$ alone.
\end{enumerate}
\end{lemma}

The above conditions ensure decidability of weak simulation as they allow to iteratively
compute the approximants and detect convergence, by the effective semilinearity of 
strong simulation over \OCN\ \cite{JKM2000}.

To obtain an upper bound for the complexity of this procedure, we will bound the sizes of all
$(S_i,S_i')$ polynomially in the sizes of $\mnet$ and $\mnet'$.
To do that, we recall some more properties of the construction, starting by describing how the nets
$\snet_i$ and $\snet_i'$ look like.

\paragraph*{The nets $\snet_i$ and $\snet_i'$}
These nets are constructed using the notion of \emph{minimal sufficient values}:
\begin{definition}
    Consider the approximant $\SIM{i}{}$ for some parameter $i$,
    which is characterised by nets $\snet_i,\snet_i'$ (cf.\ point~\ref{lem:snets:char} of
    Lemma~\ref{lem:snets} above) and let $(q,q')\in (Q\x Q')$ be a pair of states.
    By monotonicity, there is a minimal value $\suff{q,q',i}\in\N\cup\{\omega\}$ satisfying
    \begin{equation}\label{sufficient_values}
      \forall n'\in\N.\ q(\suff{q,q',i}) \notSIM{i}{} q'n'.
    \end{equation}
    Let $\suff{q,q',i}$ be $\omega$ if no finite value satisfies this condition.
\end{definition}

The idea behind the construction of nets for parameter $i+1$ is as follows.
A Simulation Game played on the arena $\snet_{i+1},\snet_{i+1}'$ mimics
the $(i+1)$-parameterized simulation game played on $\mnet,\mnet'$
until \V\ uses an $\omega$-labelled transition, leading to some game position $qn$ vs.\ $q'n'$.
Afterwards, the parameterized game would continue with the next lower parameter $i$.

By induction assumption, we can compute a representation of $\SIM{i}{}$
and hence $\suff{q,q',i}$ for every pair $(q,q')$.
Given these values, the nets $\snet_{i+1}$ and $\snet_{i+1}'$ are constructed so that instead of
making steps that are due to $\omega$-labelled transitions, \V\ can enforce the play to continue in
some subgame that he wins iff \R's counter is smaller than the hard-wired value $\suff{q,q',i}$.

This ``forcing'' of the play can be implemented for \OCN\ simulation using a standard technique
called \emph{defender's forcing} (see e.g.~\cite{KJ2006}).
So, the nets $\snet_i$ and $\snet_{i}'$ consist of the original nets $\mnet,\mnet'$ where
all $\omega$-transitions in \V's net $\mnet'$ are replaced by a small constant
defenders-forcing script, leading to the corresponding testing gadgets
that test if \R's counter is at least as large as the pre-computed sufficient value
and let \R\ win only if that is the case.

The actual test-gadgets are not very complicated: On \V's side, all gadgets are the same simple
loop over a newly introduced symbol, say $e$. Hence, $\snet_{i}'=\snet_1'$ for every $i$ and this new
net is polynomial in the size of $\mnet'$ and $\mnet$.

In \R's net $\snet_i$, the gadgets $G(q,q',i)$ for states $(q,q')$ and index $i$ solely depend on the
value $\suff{q,q',i}$:
If $\suff{q,q',i}$ is finite, it suffices to have a counter-decreasing chain of $e$-steps of length $\suff{q,q',i}$,
leading to some state which enables an action that cannot be replied to by \V.
Otherwise, if $\suff{q,q',i}=\omega$ (no counter finite value satisfies Equation~\ref{sufficient_values}),
\R\ should always lose, so a simple $e$-labelled loop can be used as gadget.
To conclude, each $\snet_i$ essentially consists of $\mnet$ plus chains $G(q,q',i)$, one for every
pair of states $(q,q')$. We summarize the crucial properties of this construction below.

\begin{lemma}\label{lem:chains}\
\begin{enumerate}
\item \label{i:omega} $\suff{q,q',1}=\omega$ for every pair $(q,q')\in Q\x Q'$.
  \item \label{i:decrease} $\suff{q,q',i}\geq \suff{q,q',i+1}$.
  \item \label{i:number} $(\snet_i,\snet_i')$ contains precisely $|Q\x Q'|$ many gadgets, each.
  \item \label{i:linear} If $\suff{q,q',i}\in \N$ then the size of gadget $G(q,q',i)$ is $\suff{q,q',i}+2$.
  \item \label{i:notransiton} No chain $G(q,q',i)$ contains transitions leading back to $\mnet$.
\end{enumerate}
\end{lemma}

Using properties \ref{i:decrease} and \ref{i:number} we derive that indeed
$(\snet_k,\snet_k')=(\snet_{k+1},\snet_{k+1}')$, and hence $\SIM{k}{}=\SIM{k+1}{}=\SIM{}{}$ for some
finite $k\in\N$.

Our goal is to bound the sizes of the nets $\snet_i,\snet_i'$ polynomially in the sizes of
$\mnet,\mnet'$ and to show that they can indeed be constructed in polynomial space.
From point~\ref{i:omega} and the fact whenever $\suff{q,q',i}=\omega$, the gadget $G(q,q',i)$ is a trivial
loop, we already know that the sizes of $\snet_1$ and $\snet_1'$ are polynomial in $\mnet,\mnet'$.
Due to the particular shape of the nets $(\snet_{i+1},\snet_{i+1}')$, it suffices to bound
the values $\suff{q,q',i}$.

\paragraph*{Bounding $\suff{q,q',i}$}
Observe that $\suff{q,q',i}$ is defined in terms of the approximant $\SIM{i}{}$, which is
characterized as the strong simulation $\SIM{}{}$ relative to the nets $\snet_i,\snet_i'$ by
Lemma~\ref{lem:snets}, Point 2. In fact, if we consider the colouring of $\SIM{}{}$
w.r.t.\ $\snet_i,\snet_i'$, the value $\suff{q,q',i}$ is the width of the belt
for $(q,q')$ if this belt is vertical and $\omega$ otherwise.
Therefore, the value $c$ in the Belt Theorem applied to this colouring bounds all finite $\suff{q,q',i}$.

We show how to bound $c$ using the sharper estimation as formulated in Section~\ref{sec:sharper:estimation}
in terms $\scc$, the maximal size of any strongly connected component and $\acyc$,
the length of the longest acyclic path in the product $\snet_i\x\snet_i'$:
\begin{equation}
    c\le \poly(\scc) + \acyc.
\end{equation}
This allows us to bound all values $\suff{q,q',i}$ and hence the size
of the nets for index $i+1$.

First, observe that the shape of all $\snet_i,\snet_i'$ (particularly Point~\ref{i:notransiton}
of Lemma~\ref{lem:chains}, and the fact that $\forall_{i\in \N} S_i'=S_1'$) implies that
the strongly connected components are unchanged from index $i=1$ onward.
Thus, $\scc$ is in fact polynomial in $\mnet,\mnet'$.
Secondly, any path in the product $\snet_i\x\snet_i'$ can be split into two (possibly empty) parts:
the part that remains in $\mnet\x\mnet'$ and a suffix that moves into at most one
gadget $G(q,q',i)$.
Since the maximal length of paths in $G(q,q',i)$ is bounded by $\suff{q,q',i}$, we can bound $\acyc$ as
follows.
\begin{equation}
    \acyc \le |\mnet\x\mnet' | + \max\{\suff{q,q',i}\in\N\ |\ (q,q')\in Q\x Q'\}.
\end{equation}
Let $W_i$ denote the maximal width of all vertical belts at level $i$,
i.e., the largest finite value $\suff{q,q',i}$ over all $(q,q')$.
By the argument above, we get for all indices $i\in\N$,
\begin{equation}
    W_{i+1} \le \poly(\mnet,\mnet') + W_i.
\end{equation}
Now, from properties \ref{i:decrease} and \ref{i:linear} of Lemma~\ref{lem:chains} we can deduce 
that there are no more than
$K=|Q\x Q'|$ indices $i$ such that $W_{i+1} \ge W_i$.
This is because the size of the value $\suff{q,q',i}$ for a particular pair $(q,q')$ can only increase once,
going from index $i$ to $i+1$ if $\suff{q,q',i}=\omega > \suff{q,q', i+1}$.
Therefore, we can bound $W_i$, and thus values $\suff{q,q',i}$, for all indices $i\in\N$ by
\begin{equation}
  W_i\leq K \cdot\ (poly(|\mnet\x\mnet'|)+ 1).
\end{equation}
We conclude that the sizes of all $\snet_i,\snet_i'$ are polynomial in the sizes of $\net$ and $\net'$.
It remains to show that we can compute these values in polynomial space, because this allows
us to effectively construct the nets for the next parameter $i+1$.
 
\paragraph*{Computing $\suff{q,q',i}$}
  We analyse the colouring of the simulation $\SIM{}{}$ relative to the one-counter nets $\snet_i$ and $\snet_i'$.
  In particular we need to answer the following questions, for each given pair of states $(q,q')\in Q\x Q'$,

  \begin{enumerate}
      \item Is the belt for $(q,q')$ vertical? And if yes,
      \item What is its exact width?
  \end{enumerate}

  By Theorem~\ref{thm:belt-theorem}, we can bound all ratios $(\rho,\rho')$, which are the
  slopes of belts polynomially. Let $(\rho,\rho')$ be the ratio of the steepest belt with $\rho'>0$.
  Recall that $c$ bounds the width of all vertical belts.
  To answer the first question, it suffices to check the colour of some point $(n,n')$
  that is both $c$-above $(\rho,\rho')$ and $c$-below of $(0,1)$, i.e., $n>c$.
  For instance, $n = c+1$ and $n'=2(c+1)(\frac{\rho}{\rho'})$ is surely such a point.

  If the belt for $(q,q')$ is vertical, then by Theorem~\ref{thm:belt-theorem}, Point 2, we have
  $qn,\notSIM{}{}q'n'$.
  Otherwise, if the belt is not vertical, then by point 1 of
  Theorem~\ref{thm:belt-theorem}, we must have $qn,\SIM{}{}q'n'$.

  To answer the second question, we consider the periodicity description of the colouring
  (cf.\ Section~\ref{sec:pspace}).
  Although this description is of exponential size and we thus cannot fully keep it in
  memory, we can, in polynomial space, compute point queries.
  Moreover, we know that the colouring in any belt is described by some non-trivial initial
  colouring and is repetitive from some exponentially bounded level onwards.
  Thus, if we consider the vertical belt for states $(q,q')$, from some level $n'_0$,
  the colouring stabilizes so that for all $n'\ge n'_0$, we have
  $qn\SIM{}{}q'n'$ iff $n< \suff{q,q',i}$.

  We can now iteratively check the colour of the point
  $(n,n'_0)$ for decreasing values $n=c$ to $0$ and some fixed, but sufficiently high $n'_0$.
  By Theorem~\ref{thm:strongsim-pspace}, this can surely be done in polynomial space.
  $\suff{q,q',i}$ must be the largest considered $n<c$ where $qn\notSIM{}{}q'n'_0$ still holds.

\end{document}